\documentclass[12pt,a4paper]{amsart}
\usepackage{amscd,amssymb,amsopn,amsmath,amsthm,graphics,amsfonts,enumerate,verbatim,calc}

\headsep=1cm \numberwithin{equation}{section}
\newtheorem{theorem}{Theorem}[section]
\newtheorem{prop}[theorem]{Proposition}

\newtheorem{corollary}[theorem]{Corollary}

\newtheorem{ex}[theorem]{Example}
\newtheorem{definition}[theorem]{Definition}

\numberwithin{equation}{section}

\begin{document}
\pagenumbering{gobble}
\title{F-index and coindex of some derived graphs}
\author{Nilanjan De}

\address{Department of Basic Sciences and Humanities, Calcutta Institute of Engineering and Management,
Kolkata, India}
 \email{ de.nilanjan@rediffmail.com}



\begin{abstract}
In this study, the explicit expressions for F-index and coindex of derived graphs such as a line graph, subdivision graph, vertex-semitotal graph, edge-semitotal graph, total graph and paraline graph (line graph of the subdivision graph) are obtained. \vspace{0.2 cm} 
\end{abstract}

\maketitle

\section{Introduction}
 Throughout the paper, we consider finite, connected and undirected graphs without any self-loops or multiple edges. Let $G$ be such a graph with vertex set $V(G)$ and edge set $E(G)$. Also let, $n$ and $m$ be the number of vertices and edges of $G$ and the edge connecting the vertices $u$ and $v$ is denoted by $uv$. Let ${{d}_{G}}(v)$ denote the degree of the vertex $v$ in $G$ which is the number of edges incident to $v$, that is, the number of first neighbors of $v$. Topological indices are numeric quantity derived from a molecular graph which correlate the physico-chemical properties of the molecular graph and have been found to be useful in isomer discrimination, quantitative structure-activity relationship (QSAR) and structure-property relationship (QSPR)  and are necessarily invariant under automorphism of graphs.

The first and the second classical Zagreb index of $G$ denoted by ${{M}_{1}}(G)$ and ${{M}_{2}}(G)$ respectively are one of the oldest topological indices introduced in \cite{gutm72} by Gutman and Trinajsti\'{c} and defined as
\[{{M}_{1}}(G)=\sum\limits_{v\in V(G)}{{{d}_{G}}{{(u)}^{2}}}=\sum\limits_{uv\in E(G)}{\left[ {{d}_{G}}(u)+{{d}_{G}}(v) \right]}\]        and    \[{{M}_{2}}(G)=\sum\limits_{uv\in V(G)}{{{d}_{G}}(u){{d}_{G}}(v)}.\]

These indices are most important topological indices in study of structure property correlation of molecules and have received attention in mathematical as well as chemical literature and have been extensively studied both with respect to mathematical and chemical point of view (see \cite{zho05,zho07,kha09,xu15,das15}).

Another topological index, defined as sum of cubes of degrees of all the vertices was also introduced in the same paper, where the first and second Zagreb indices were introduced \cite{gutm72}. Furtula and Gutman in \cite{fur15} recently investigated this index and named this index as ``forgotten topological index" or ``F-index" and showed that the predictive ability of this index is almost similar to that of first Zagreb index and for the entropy and acetic factor, both of them yield correlation coefficients greater than 0.95. The F-index of a graph $G$ is defined as

\[F(G)=\sum\limits_{v\in V(G)}{{{d}_{G}}{{(u)}^{3}}}=\sum\limits_{uv\in E(G)}{\left[ {{d}_{G}}{{(u)}^{2}}+{{d}_{G}}{{(v)}^{2}} \right]}.\]

Recently, the concept of F-index attracting much attention of researchers. The present author studied this index for different graph operations \cite{de16a} and also  studied F-index of several classes of nanostar dendrimers and total transformation graphs in \cite{de16c} and \cite{de16d}. In \cite{abd15}, Abdoa et al. investigate the trees extremal with respect to the F-index. Analogous to Zagreb coindices,  the present author  introduced  the  F-coindex  in \cite{de16b}. Thus, the F-coindex of a graph $G$ is defined as

\[\bar{F}(G)=\sum\limits_{uv\in E(\bar{G})}{\left[ {{d}_{G}}{{(u)}^{2}}+{{d}_{G}}{{(v)}^{2}} \right]}.\]

Different topological indices of some derived graphs such as a line graph, subdivision graph, vertex-semitotal graph, edge-semitotal graph, total graph and paraline graph have already been studied by many researcher. Gutman et al. in \cite{gut15} found first Zagreb index of some derived graphs. Also, Basavanagoud et al. in \cite{bas16} and \cite{bas15} calculated  multiplicative Zagreb indices and second Zagreb index of some derived graphs. In this paper, we continue the previous work to determine the F-index of these derived graphs. Throughout this paper, as usual, ${{C}_{n}}$ and ${{S}_{n}}$ denote the cycle and star graphs on $n$ vertices.

\section{Main results}

In this section, first we define different subdivision-related graphs and state some relevant results.

Line graph $L=L(G)$ is the graph with vertex set $V(L)=E(G)$ and whose vertices correspond to the edges of $G$ with two vertices being adjacent if and only if the corresponding edges in $G$ have a vertex in common two.

Subdivision graph $S=S(G)$ is the graph obtained from $G$ by replacing each of its edges by a path of length two, or equivalently, by inserting an additional vertex into each edge of $G$.

Vertex-semitotal graph ${{T}_{1}}={{T}_{1}}(G)$ with vertex set $V(G)\cup E(G)$ and edge set $E(S)\cup E(G)$ is the graph obtained from $G$ by adding a new vertex corresponding to each edge of $G$ and by joining each new vertex to the end vertices of the edge corresponding to it.

Edge-semitotal graph ${{T}_{2}}={{T}_{2}}(G)$ with vertex set $V(G)\cup E(G)$ and edge set $E(S)\cup E(L)$ is the graph obtained from $G$ by inserting a new vertex into each edge of $G$ and by joining with edges those pairs of these new vertices which lie on adjacent edges of $G$.

The Total graph of a graph $G$ is denoted by $T=T(G)$ with vertex set $V(G)\cup E(G)$ and any two vertices of $T(G)$ are adjacent if and only if they are either incident or adjacent in $G$ Total graph.

The Paraline graph $PL=PL(G)$ is the line graph of the subdivision graph with $2m$ vertices.
For details definitions of different derived graphs we refer our reader to \cite{gut15}.

\subsection{F-index of derived graphs}

In order to calculate the first F-index of the above specified derived graphs, we need following graph invariants. One of the redefined versions of Zagreb index is given by
	\[\operatorname{Re}Z{{G}_{3}}(G)=\sum\limits_{uv\in E(G)}{{{d}_{G}}(u){{d}_{G}}(v)[{{d}_{G}}(u)+{{d}_{G}}(v)]}.\]
For different recent study of these index see \cite{gao16,de16d}. In this paper we use another index to express the F-index of different derived graphs of a graph G and is denoted by ${{\xi }_{4}}(G)$, which is defined as

	\[\sum\limits_{v\in V(G)}{{{d}_{G}}{{(v)}^{4}}}=\sum\limits_{uv\in E(G)}{[{{d}_{G}}{{(u)}^{3}}+{{d}_{G}}{{(v)}^{3}}]}={{\xi }_{4}}(G).\]
Now in the following we compute the F-index of the above specified derived graphs.
\begin{prop} \label{p1}
Let G be be graph of order n and size m, then
\[F(L)={{\xi }_{4}}(G)+3\operatorname{Re}Z{{G}_{3}}(G)-6F(G)-12{{M}_{2}}(G)+12{{M}_{1}}(G)-8m.\]
\end{prop}
\begin{proof}For a Line graph, any two vertices are adjacent if the corresponding edges of $G$ are incident with a common vertex. Since, the edge $uv$ of the graph $G$ is incident to $[{{d}_{G}}(u)+{{d}_{G}}(v)-2]$ other edges of $G$, we have
\begin{eqnarray*}
F(L)&=&\sum\limits_{uv\in E(G)}{{{[{{d}_{G}}(u)+{{d}_{G}}(v)-2]}^{3}}}\\
    &=&\sum\limits_{uv\in E(G)}{[{{d}_{G}}{{(u)}^{3}}+{{d}_{G}}{{(v)}^{3}}]}+3\sum\limits_{uv\in E(G)}{{{d}_{G}}(u){{d}_{G}}(v)[{{d}_{G}}(u)+{{d}_{G}}(v)]}\\
    &&-6\sum\limits_{uv\in E(G)}{[{{d}_{G}}{{(u)}^{2}}+{{d}_{G}}{{(v)}^{2}}]}-12\sum\limits_{uv\in E(G)}{{{d}_{G}}(u){{d}_{G}}(v)}\\
    &&+12\sum\limits_{uv\in E(G)}{[{{d}_{G}}(u)+{{d}_{G}}(v)]}-8m\\
    &=&{{\xi }_{4}}(G)+3\operatorname{Re}Z{{G}_{3}}(G)-6F(G)-12{{M}_{2}}(G)+12{{M}_{1}}(G)-8m.
\end{eqnarray*}
Hence the desired result follows.\end{proof}

\begin{prop} \label{p2}
Let G be be graph of order n and size m, then
\[F(S)=F(G)+8m\]
\end{prop}
\begin{proof}
Since for $u\in V(S)\cap V(G)$,  ${{d}_{S}}(u)={{d}_{G}}(u)$ and for $e=uv\in V(S)\cap E(G)$,  ${{d}_{S}}(e)=2$, we have

\medskip
$F(S)=\sum\limits_{u\in V(G)}{{{d}_{G}}{{(u)}^{3}}}+\sum\limits_{uv\in E(G)}{{{2}^{3}}}=F(G)+8m.$
\end{proof}
\begin{prop} \label{p3}
Let G be be graph of order n and size m, then
$F({{T}_{1}})=8F(G)+8m$.
\end{prop}
\begin{proof}
Since for $u\in V({{T}_{1}})\cap V(G)$,  ${{d}_{{{T}_{1}}}}(u)=2{{d}_{G}}(u)$ and for $e=uv\in V({{T}_{1}})\cap E(G)$,  ${{d}_{{{T}_{1}}}}(e)=2$, we have

\medskip
$F({{T}_{1}})=$$\sum\limits_{u\in V(G)}{{{[2{{d}_{G}}(u)]}^{3}}}+\sum\limits_{uv\in E(G)}{{{2}^{3}}}=8F(G)+8m$.
\end{proof}

\begin{prop} \label{p4}
Let G be be graph of order n and size m, then
\[F({{T}_{2}})=F(G)+{{\xi }_{4}}(G)+3\operatorname{Re}Z{{G}_{3}}(G).\]
\end{prop}
\begin{proof}
Since for $u\in V({{T}_{2}})\cap V(G)$,  ${{d}_{{{T}_{1}}}}(u)={{d}_{G}}(u)$ and for $e=uv\in V({{T}_{2}})\cap E(G)$,  ${{d}_{{{T}_{1}}}}(e)={{d}_{G}}(u)+{{d}_{G}}(v)$, we have
\begin{eqnarray*}
F({{T}_{2}})&=&\sum\limits_{u\in V(G)}{{{d}_{G}}{{(u)}^{3}}}+\sum\limits_{uv\in E(G)}{{{[{{d}_{G}}(u)+{{d}_{G}}(v)]}^{3}}}\\
            &=&F(G)+\sum\limits_{uv\in E(G)}{[{{d}_{G}}{{(u)}^{3}}+{{d}_{G}}{{(v)}^{3}}]}+3\sum\limits_{uv\in E(G)}{[{{d}_{G}}(u)+{{d}_{G}}(v)]}{{d}_{G}}(u){{d}_{G}}(v)\\
            &=&F(G)+{{\xi }_{4}}(G)+3\operatorname{Re}Z{{G}_{3}}(G).
\end{eqnarray*}
Hence the result.\end{proof}

\begin{prop}\label{p5}
Let G be be graph of order n and size m, then
\[F(T)=8F(G)+{{\xi }_{4}}(G)+3\operatorname{Re}Z{{G}_{3}}(G).\]
\end{prop}

\begin{proof}
Since for $u\in V(T)\cap V(G)$,  ${{d}_{T}}(u)=2{{d}_{G}}(u)$ and for $e=uv\in V(T)\cap E(G)$,  ${{d}_{{{T}_{1}}}}(e)={{d}_{G}}(u)+{{d}_{G}}(v)$, we have
\begin{eqnarray*}
F(T)&=&\sum\limits_{u\in V(G)}{{{[2{{d}_{G}}(u)]}^{3}}}+\sum\limits_{uv\in E(G)}{{{[{{d}_{G}}(u)+{{d}_{G}}(v)]}^{3}}}\\
    &=&8F(G)+\sum\limits_{uv\in E(G)}{[{{d}_{G}}{{(u)}^{3}}+{{d}_{G}}{{(v)}^{3}}]}+3\sum\limits_{uv\in E(G)}{[{{d}_{G}}(u)+{{d}_{G}}(v)]}{{d}_{G}}(u){{d}_{G}}(v)\\
    &=&8F(G)+{{\xi }_{4}}(G)+3\operatorname{Re}Z{{G}_{3}}(G).
\end{eqnarray*}
Hence the desired result follows.  \end{proof}

\begin{prop} \label{p6}
Let G be graph of order n and size m, then
$F(PL)={{\xi }_{4}}(G)$.
\end{prop}
\begin{proof}
Since, for the paraline graph PL,  ${{d}_{G}}(u)$ of its vertices have the same degree as the vertex $u$ of the graph $G$ and paraline graph PL has $2m$ vertices, we have

\medskip
$F(PL)=\sum\limits_{x\in V(G)}{{{d}_{PL}}{{(x)}^{3}}}=\sum\limits_{u\in V(G)}{{{d}_{G}}(u){{[{{d}_{G}}(u)]}^{3}}}={{\xi }_{4}}(G).$
\end{proof}

\begin{ex}
Consider the cycle ${{C}_{n}}$ with $n$ vertices where every vertex is of degree 2, then\\
(i) $F(L({{C}_{n}}))=8n$,
(ii) $F(S({{C}_{n}}))=16n$,
(iii) $F({{T}_{1}}({{C}_{n}}))=72n$,
(iv) $F({{T}_{2}}({{C}_{n}}))=72n$,
(v) $F(T({{C}_{n}}))=128n$,
(vi) $F(PL({{C}_{n}}))=16n.$
\end{ex}

\begin{ex}
Consider the cycle ${{S}_{n}}$ with $n$ vertices, then\\
(i) $F(L({{S}_{n}}))=8n$,
(ii) $F(S({{S}_{n}}))=(n-1)(n^2-2n+3)$,
(iii) $F({{T}_{1}}({{S}_{n}}))=72n$,
(iv) $F({{T}_{2}}({{S}_{n}}))=72n$,
(v) $F(T({{S}_{n}}))=128n$,
(vi) $F(PL({{S}_{n}}))=16n.$
\end{ex}

\subsection{F-coindex of derived graphs}

The F-index is the sum over the adjacent edges and F-coindex is the sum of the contribution of non adjacent pair of vertices.  The concept of F-coindex was introduced by De et al. \cite{de16b} and have shown that the F-coindex can predict the octanol water partition coefficients of molecular structures very efficiently. In that paper the following proposition was proved, which is necessary in the following study.

\begin{prop} \label{p7}
Let G be a simple graph with n vertices and m edges, then
	\[\bar{F}(G)=(n-1){{M}_{1}}(G)-F(G).\]
\end{prop}

The following proposition was proved in \cite{gut15} and also required here.

\begin{prop} \label{p8}
Let G be a graph of order n and size m, then\\
${{M}_{1}}(L)=F(G)-4{{M}_{1}}(G)+2{{M}_{2}}(G)+4m$\\
${{M}_{1}}(S)={{M}_{1}}(G)+4m$\\
${{M}_{1}}({{T}_{1}})=F(G)+{{M}_{1}}(G)+2{{M}_{2}}(G)$\\
${{M}_{1}}({{T}_{2}})=4{{M}_{1}}(G)+4m$\\
${{M}_{1}}(T)=F(G)+4{{M}_{1}}(G)+2{{M}_{2}}(G)$\\
${{M}_{1}}(PL)=F(G).$
\end{prop}

Now we calculate the F-coindex of the different derived graphs.

\begin{prop}
Let G be a graph of order n and size m, then
\[\bar{F}(L)=(m+5)F(G)-4(m+2){{M}_{1}}(G)+2(m+5){{M}_{2}}(G)-{{\xi }_{4}}(G)-3\operatorname{Re}Z{G}_{3}(G)+4m(m+1).\]
\end{prop}

\begin{proof}
Since, the line graph L has m vertices, so using propositions \ref{p7}, \ref{p8} and \ref{p1}, we get
\begin{eqnarray*}
\bar{F}(L)&=&(m-1){{M}_{1}}(L)-F(L)\\
          &=&(m-1)[F(G)-4{{M}_{1}}(G)+2{{M}_{2}}(G)+4m]-[{{\xi }_{4}}(G)+3\operatorname{Re}Z{{G}_{3}}(G)\\
          &&-6F(G)-12{{M}_{2}}(G)+12{{M}_{1}}(G)-8m]\\
          &=&(m-1)F(G)+6F(G)-4(m-1){{M}_{1}}(G)-12{{M}_{1}}(G)+2(m-1){{M}_{2}}(G)\\
          &&+12{{M}_{2}}(G)+4m(m+1)-{{\xi }_{4}}(G)-3\operatorname{Re}Z{{G}_{3}}(G)
\end{eqnarray*}
from where the desired result follows. \end{proof}

\begin{prop}
Let G be a graph of order n and size m, then
\[\bar{F}(S)=(m+n-1){{M}_{1}}(G)-F(G)+4m(m+n-3).\]
\end{prop}

\begin{proof}
Since, the subdivision graph S has (m+n) vertices, so using propositions \ref{p7}, \ref{p8} and \ref{p2}, we get
\begin{eqnarray*}
\bar{F}(S)&=&(m+n-1){{M}_{1}}(S)-F(S)\\
          &=&(m+n-1)({{M}_{1}}(G)+4m)-F(G)-8m\\
          &=&(m+n-1){{M}_{1}}(G)-F(G)+4m(m+n-1)-8m.
\end{eqnarray*} \end{proof}

\begin{prop}
Let G be a graph of order n and size m, then
\[\bar{F}({{T}_{1}})=4(m+n-1){{M}_{1}}(G)-8F(G)+4m(m+n-3).\]
\end{prop}

\begin{proof}
Since, the total graph ${{T}_{1}}$  has $(m+n)$ vertices, so using propositions \ref{p7}, \ref{p8} and \ref{p3}, we get
\begin{eqnarray*}
\bar{F}({{T}_{1}})&=&(m+n-1){{M}_{1}}({{T}_{1}})-F({{T}_{1}})\\
                  &=&(m+n-1)(4{{M}_{1}}(G)+4m)-8F(G)-8m\\
                  &=&4(m+n-1){{M}_{1}}(G)-8F(G)+4m(m+n-3)
\end{eqnarray*}
Hence we get the desired result. \end{proof}

\begin{prop}
Let G be a graph of order n and size m, then
\[\bar{F}({{T}_{2}})=(m+n-2)F(G)+(m+n-1){{M}_{1}}(G)+2(m+n-1){{M}_{2}}(G)-{{\xi }_{4}}(G)-3\operatorname{Re}Z{{G}_{3}}(G).\]
\end{prop}

\begin{proof}
Since, the total graph ${{T}_{2}}$ has $(m+n)$ vertices, so using propositions \ref{p7}, \ref{p8} and \ref{p4}, we get
\begin{eqnarray*}
\bar{F}({{T}_{2}})&=&(m+n-1){{M}_{1}}({{T}_{2}})-F({{T}_{2}})\\
                  &=&(m+n-1)(F(G)+{{M}_{1}}(G)+2{{M}_{2}}(G))-F(G)-{{\xi }_{4}}(G)-3\operatorname{Re}Z{{G}_{3}}(G)\\
                  &=&(m+n-2)F(G)+(m+n-1){{M}_{1}}(G)+2(m+n-1){{M}_{2}}(G)-{{\xi }_{4}}(G)\\
                  &&-3\operatorname{Re}Z{{G}_{3}}(G).
\end{eqnarray*}
Hence the required result follows. \end{proof}

\begin{prop}
Let G be a graph of order n and size m, then
\[\bar{F}(T)=(m+n-9)F(G)+4(m+n-1){{M}_{1}}(G)+2(m+n-1){{M}_{2}}(G)-{{\xi }_{4}}(G)-3\operatorname{Re}Z{{G}_{3}}(G).\]
\end{prop}

\begin{proof}
Since, the total graph $T$ has $(m+n)$ vertices, so using propositions \ref{p7}, \ref{p8} and \ref{p5}, we get
\begin{eqnarray*}
\bar{F}(T)&=&(m+n-1){{M}_{1}}(T)-F(T)\\
          &=&(m+n-1)(F(G)+4{{M}_{1}}(G)+2{{M}_{2}}(G))-8F(G)-{{\xi }_{4}}(G)-3\operatorname{Re}Z{{G}_{3}}(G)\\
          &=&(m+n-9)F(G)+4(m+n-1){{M}_{1}}(G)+2(m+n-1){{M}_{2}}(G)-{{\xi }_{4}}(G)\\
          &&-3\operatorname{Re}Z{{G}_{3}}(G).
\end{eqnarray*}
Hence the result follows. \end{proof}

\begin{prop}
Let G be a graph of order n and size m, then
\[\bar{F}(PL)=(2m-1)F(G)-{{\xi }_{4}}(G)\]
\end{prop}

\begin{proof}
Since, the paraline graph $PL$ has $2m$ vertices, so using propositions \ref{p7}, \ref{p8} and \ref{p6}, we get
\medskip
\[\bar{F}(PL)=(2m-1){{M}_{1}}(PL)-F(PL)=(2m-1)F(G)-{{\xi }_{4}}(G).\]
\end{proof}

\begin{ex}
Consider the cycle ${{C}_{n}}$ with $n$ vertices, then\\
(i) $\bar{F}(L({{C}_{n}}))=4n(n-3)$,
(ii) $\bar{F}(S({{C}_{n}}))=8n(2n-3)$,
(iii) $\bar{F}({{T}_{1}}({{C}_{n}}))=4n(10n-23)$,
(iv) $\bar{F}({{T}_{2}}({{C}_{n}}))=4n(10n-23),$
(v) $\bar{F}(T({{C}_{n}}))=32n(2n-5)$,
(vi) $\bar{F}(PL({{C}_{n}}))=8n(2n-3).$
\end{ex}

\begin{ex}
Consider the star graph ${{S}_{n}}$ with $n$ vertices, then\\
(i) $\bar{F}(L({{S}_{n}}))=0$,\\
(ii) $\bar{F}(S({{S}_{n}}))=(n-1)(n^2+8n-18)$,\\
(iii) $\bar{F}({{T}_{1}}({{S}_{n}}))=16(n-1)(n-2)$,\\
(iv) $\bar{F}({{T}_{2}}({{S}_{n}}))=(n-1)(n^3-n^2-12n-2),$\\
(v) $\bar{F}(T({{S}_{n}}))=(n-1)(6n^3-n^4-11n^2+14n-16)$,\\
(vi) $\bar{F}(PL({{S}_{n}}))=(n-1)(n^3-4n^2+7n-6).$
\end{ex}
\section{Conclusion}

In this paper, we have studied the F-index and coindex of different derived graphs. For further study, F-index and coindex of some other derived and composite graphs can be computed.

\end{document}